\documentclass{sig-alternate}
\title{Indexing Reverse Top-k Queries}

\numberofauthors{1}

\author{
\alignauthor Sean Chester, Alex Thomo, S. Venkatesh, and Sue Whitesides\\
       \affaddr{Computer Science Department}\\
       \affaddr{University of Victoria}\\
       \affaddr{PO Box 1700 STN CSC}\\
       \affaddr{Victoria, Canada}\\
       \email{\{schester, sue\}@uvic.ca, \{thomo, venkat\}@cs.uvic.ca}
}
\date{\today}

\newcommand{\ea}{~et~al.}

\newcommand{\vvec}{\textbf{v}}

\newcommand{\topk}{\mathrm{top}$-$k}
\newcommand{\mrtop}{maxRTOP}
\newcommand{\dd}{top-$k$ rank depth}

\newcommand{\ddf}{depth}

\newcommand{\origin}{O}

\newtheorem{theorem}{Theorem}[section]
\newtheorem{lemma}[theorem]{Lemma}
\newtheorem{corollary}[theorem]{Corollary}
\newtheorem{proposition}[theorem]{Proposition}
\newtheorem{definition}{Definition}

\usepackage{balance}
\usepackage{amsmath, amssymb}
\usepackage{algorithm, algorithmic, float}
\usepackage{graphicx}
\usepackage[colorlinks=false]{hyperref}
\hypersetup{
 pdfauthor = {Sean Chester, Alex Thomo, S. Venkatesh, and Sue Whitesides},
 pdftitle = {Indexing Reverse Top-k Queries},
 pdfsubject = {Maximal Reverse Top-k Queries},
 pdfkeywords = {reverse $\topk$, top k depth, arrangements of lines, access methods}
}

\graphicspath{.}

\begin{document}
\maketitle

\begin{abstract}
We consider the recently introduced monochromatic reverse $\topk$ queries which asks for, given a new tuple $q$ and a dataset $\mathcal{D}$, all possible $\topk$ queries on $\mathcal{D}\cup\{q\}$ for which $q$ is in the result.  Towards this problem, we focus on designing indexes in two dimensions for repeated (or batch) querying, a novel but practical consideration.  We present the novel insight that by representing the dataset as an arrangement of lines, a critical $k$-polygon can be identified and used exclusively to respond to reverse $\topk$ queries.  We construct an index based on this observation which has guaranteed worst-case-logarithmic query cost.

We implement our work and compare it to related approaches, demonstrating that our index is fast in practice.  Furthermore, we demonstrate through our experiments that a $k$-polygon is comprised of a small proportion of the original data, so our index structure consumes little disk space.
\end{abstract}




\section{Introduction}
Imagine a software engineering team in the early stages of developing a new single-player console game.  Given aggressive timelines until product launch, they need to prioritise the development efforts.  Market research reveals that games in this category are typically assessed by end-users in terms of the quality of the graphics and the intelligence of the AI.  Furthermore, different users express different trade-offs in terms of which of these two metrics they believe to be the most important.  

If most games in this category have focused on the development of graphics, then the graphics-focused market is perhaps saturated and the engineering team may encounter more success by focusing instead on the development of the AI.  More broadly stated, the exposure, and indeed success, of a product depends largely on how well it ranks against other, similar products.  The task with which the engineering team is faced here, in fact, is to assess how to maximise ``$\topk$ exposure'', the breadth of $\topk$ queries for which their product is returned.

Computing the $\topk$ exposure of a product is the objective of a reverse $\topk$ query, introduced recently by Vlachou\ea~\cite{VlachouRTOP}.  A {\em traditional} (linear) $\topk$ query is a weight vector $\left<w_1, w_2\right>$ that assigns a weight to each attribute of the relation $\mathcal{D}$.  The result set contains the $k$ tuples $\left(a_1, a_2\right)\in\mathcal{D}$ for which $w_1*a_1+w_2*a_2$ is highest.  A {\em reverse} $\topk$ query, given as input a numerical dataset $\mathcal{D}$, a value $k$, and a new query tuple $q$, reports the set of {\em traditional} $\topk$ queries on $\mathcal{D}\cup\{q\}$ for which $q$ is in the result set.  

In this paper, we focus on two dimensions and on the version of the problem in which the infinitely many possible traditional $\topk$ queries are considered.\footnote{This version of the problem was termed a {\em monochromatic} reverse $\topk$ query by Vlachou\ea.  The alternative is the {\em bichromatic} version in which the traditional $\topk$ queries to be considered are limited to those enumerated in a finite relation.}
Consequently here, the result of a {\em reverse} $\topk$ query is an infinite set of weight vectors, which throughout this paper we assume to be represented as a set of disjoint angular intervals.  For example, the angular interval $\left(\pi/6, \pi/4\right)$ describes the infinitely many weight vectors with an angular distance from the positive $x$-axis between $\pi/6$ and $\pi/4$, exclusive.  

\subsubsection*{A Broader Perspective}
The processing of a reverse $\topk$ query is in itself interesting, but it is also important to consider where it fits within the context of a broader workflow.  That is to say, what prompts the query and what occurs after the query is executed affects how the query should be processed.  This consideration of broader context motivates our work. 

A single reverse $\topk$ query executed on its own is informative but not very actionable.  A more likely scenario is that an analyst is trying to compare the impact of {\em many} product options in order to gauge which might be the most successful among them.  In the case of the game development scenario, it is more useful for the engineering team to evaluate {\em many} trade-offs between AI and graphics in order to compare what degrees of relative prioritisation will make the game stand out to the broadest range of end-users.  

With this in mind, we propose the first indexing-based solution to reverse $\topk$ queries.  The computational advantage of this approach is that the majority of the cost can be absorbed before the queries arrive.  In contrast, the existing techniques (of Vlachou\ea~\cite{VlachouRTOP} and of Wang\ea~\cite{wangMRTOP}) inherently depend on knowledge of the current query, so the linear-cost computation must be restarted for each of the many queries in a batch.

It is crucial to consider also what becomes of the output for a reverse $\topk$ query.  If it is meant for direct human consumption, then the end-user can only interpret a succinct representation.  Fragmenting output intervals into many sub-intervals would overwhelm the user.  
The argument we make here is that not every correct output is equivalent.  In particular, the (sometimes severe and unsorted) fragmentation of output intervals by existing techniques is quite undesirable.

To address the importance of output, we define {\em maximal} reverse top-k (\mrtop) queries, in which adjacent output intervals of a reverse $\topk$ query must be merged.  In this sense, each reported interval in the solution is {\em maximal}.  A nice property of our index-based approach is that it naturally produces this higher quality, maximal output.

\subsubsection*{Our Approach}
Our approach to the \mrtop\ problem is to consume the cost of sorting an approximation of and conducting a plane sweep on the $k$-skyband of $\mathcal{D}$ (a subset of $\mathcal{D}$ that we describe in Definition~\ref{def:skyband}) in order to design an index with query cost guaranteed to be logarithmic in the size of the true $k$-skyband.  

We achieve this through novel geometric insight into the problem.  Conceptually, we transform each tuple of $\mathcal{D}$ into a line in Euclidean space, constructing an {\em arrangement of lines} (i.e., set of vertices, edges, and faces based on their intersections; see Definition~\ref{def:arrangement}).  From that arrangement, we show that a critical star-shaped polygon can be extracted for each value of $k$.  The importance of this polygon is that if we apply the same transform to the new query tuple $q$ to produce a line $l_q$, then the \mrtop\ response is exactly the intersection of $l_q$ with the interior of this polygon.  So, with this insight, the challenge becomes to effectively index the polygon.  

A crucial observation that we derive is that the polygon has a particular form: in a convex approximation, the endpoints of any edge must be within $\mathcal{O}(k)$ edges in the original polygon.  Leveraging this insight permits our producing an index with guaranteed $\mathcal{O}(\log{n})$ query cost.

Computationally, the construction and representation of an arrangement of lines is somewhat expensive.  Instead, we demonstrate that the only tuples that could form part of the critical polygon are those among the $k$-skyband of $\mathcal{D}$.  So, we approximate the $k$-skyband, sort these lines based on their x-intercept (the dominating cost of the algorithm), and introduce a radial plane sweep algorithm to build the polygon index.

Our query algorithm is a binary search on the convexified polygon.  The recursion is based on the slope of the query line compared to the convex hull at the recursion point.  Once we discover the at most two intersections of the query line with the convex hull, we perform at most two $\mathcal{O}(k)$ sequential scans to derive the exact solution.  We can efficiently process a batch of queries, because we need only intersect the transformed line for each query with the same star-shaped polygon using the same index structure.

\subsubsection*{Summary of Our Contributions}
To the MRTOP problem, we make several substantial contributions:
\begin{itemize}
 \item{The definition of {\em maximal} reverse $\topk$ (\mrtop) queries, which accounts for the neglected post-processing that invariably must be done on the result set and that illustrates a major weakness in techniques that produce highly fragmented and uninterpretable result sets.} 
 \item{The first index-based approach to reverse $\topk$ queries, leveraging our thorough geometric analysis of the problem and our resultant novel insight into the problem properties.  This approach is the first to produce logarithmic query cost for reverse $\topk$ queries.}
 \item{The creation, optimisation, and publishing of comparable implementations for the work of Vlachou\ea, Wang\ea, and us, and an experimental evaluation that compares each of the three algorithms for different data distributions.}
\end{itemize}

\section{Literature}\label{sec:literature}
Monochromatic reverse $\topk$ queries are quite new, introduced by Vlachou\ea~\cite{VlachouRTOP} and an example of the growing field of Reverse Data Management~\cite{gatterbauer_rdm}.  As yet, there are two algorithms (besides ours) to efficiently answer monochromatic reverse $\topk$ queries, the one originally proposed by Vlachou\ea~\cite{VlachouRTOP}, and a subsequent algorithm proposed by Wang\ea~\cite{wangMRTOP}.  Both are linear-cost, two-dimensional algorithms.  

The algorithm of Vlachou\ea, refined in their more recent work~\cite{VlachouTKDE}, interprets each tuple as a point in Euclidean space and relies on the pareto-dominance relationships between the query point $q$ and the points in $\mathcal{D}$.  In particular it groups the points of $\mathcal{D}$ into those that dominate $q$, are dominated by $q$, and are incomparable to $q$.  The second phase is to execute a radial plane sweep over the set of points that are incomparable to $q$ in order to derive the exact solution.  Given the nature of the algorithm, we believe an interesting research direction may be to incorporate the work of Zou~and~Chen~\cite{pdg} on the Pareto-Based Dominant Graph.

Das\ea~\cite{das} describe a duality transform approach for traditional $\topk$ queries.  Wang\ea\ adapt this work into an algorithm that maintains a list of segments of $l_q$ as follows.  First they transform the query into a dual line $l_q$.  Then, for each tuple $p$ in $\mathcal{D}$, they construct the dual line $l_p$ and split $l_q$ at its intersection point with $l_p$.  For each segment of $l_q$, they maintain how many of the tuples in $\mathcal{D}$ so far have a higher rank than $q$ over that segment, discarding the segment as soon as the count exceeds $k-1$.  Their work reports an experimental order-of-magnitude improvement over that of Vlachou\ea, a claim that our experiments independently verify.

The foremost distinctions (other than techniques) of our work from these (\cite{VlachouRTOP,wangMRTOP}) is, first, that the majority of our computation is independent of $q$ (query-agnostic), and, second, the recognition that the queries more plausibly are executed in batches.  Together, these distinctions legitimise the construction of an asymptotically faster index-based approach.  We also note that Wang\ea\ propose a cubic-space rudimentary index that materialises the solution to every query, but which cannot handle the case when $q\not\in\mathcal{D}$.

Our approach in this paper, enabled by our earlier research on threshold queries~\cite{ChesterVP}, is based on arrangements, a central concept in computational geometry.  We suggest the 1995 survey by Sharir~\cite{SharirSurvey} for the interested reader.  It is particularly relevant because of its discussion of the computation of zones in an arrangement (i.e., the set of cells intersected by a surface).  The \textit{de facto} standard for representing arrangements is the \textit{doubly connected edge list}, which is detailed quite well in the introductory text of de~Berg\ea~\cite{compGeoText}.  

Our analysis of the arrangement is centred around data depth and depth contours.  Within Statistics, data depth is a well studied approach to generalising concepts like \textit{mean} to higher dimensions and a number of different depth measures were recently evaluated against each other by Hugg\ea~\cite{Tufts}.  Top-$k$ rank depth has not been studied, but is similar to arrangement depth, which is investigated by Rousseeuw and Hubert~\cite{arrangeDepth}, particularly with regard to bounding and algorithmically computing the maximum depth of a point within an arrangement.  It is important to note, however, that we deviate from these other concepts of data depth by setting the face containing the origin, rather than the external face, to have minimal depth and by not ensuring affine equivariance.  As a consequence, we cannot make the assertion about connectedness and monotonicity offered by the study of depth contours by Zuo~and~Serfling~\cite{zuo}.  

A last comment about related work pertains to literature on the traditional $\topk$ query problem, surveyed by Ilyas\ea~\cite{ilyasSurvey}.  Results in that domain cannot be straightforwardly applied here, as argued by Vlachou\ea, because non-null solutions to a monochromatic reverse $\topk$ query are infinite sets. 
\section{Preliminaries}\label{sec:preliminaries}
In this section we formally introduce the problem under study and define the scaffolding upon which this work relies.  

Throughout all this work, we assume queries are executed on a two-dimensional, numeric relation $\mathcal{D}$ which is a set of tuples $(a_1\in\mathbb{R}, a_2\in\mathbb{R})$.  Tuples can also alternatively be conceived as points $(a_1,a_2)$ in the Euclidean plane or as two-dimensional vectors $\left<a_1,a_2\right>$.  We assume $|\mathcal{D}|$ is ``large'', and that $k\in\mathbb{Z}^{+}<<|\mathcal{D}|$.

To begin, a {\em traditional, linear $\topk$ query} is a pair of weights $w_1, w_2$.  The response is the set of $k$ tuples in $\mathcal{D}$, which, when interpeted as vectors, have the largest dot product with $\left<w_1, w_2\right>$.  That is to say:

\begin{definition}\label{def:topk}
The response to a {\em traditional, linear $\topk$ query}, $\vec{w}=\left<w_1, w_2\right>$, is the set: $$TOP(\vec{w})=\{\vec{v}\in\mathcal{D}:|\{\vec{u}\in\mathcal{D}:\vec{u}\cdot\vec{w}>\vec{v}\cdot\vec{w}\}|<k\}.$$
\end{definition}

The {\em monochromatic reverse $\topk$ query}, introduced by Vlachou\ea~\cite{VlachouRTOP}, which we refer to simply as a {\em reverse $\topk$ query} in this paper, is a tuple $q=\left(q_1,q_2\right)$ not necessarily in $\mathcal{D}$.  The response is the set of traditional, linear $\topk$ queries on $\mathcal{D}\cup\{q\}$ for which $q$ is in the result set.  Formally:

\begin{definition}\label{def:rtop}
The response to a {\em reverse $\topk$ query}, $q=\left(q_1,q_2\right)$, is the set of angles 
\begin{eqnarray*}
RTOP(q)&=&\{\theta\in[0,\pi/2]:\\
&&|\{v\in\mathcal{D}:v_1+v_2\tan{\theta}>q_1+q_2\tan{\theta}\}|<k\}.\\
\end{eqnarray*}
\end{definition}

We introduce now a more user-conscious problem definition, that of a {\em maximal reverse $\topk$ query}.  The response to $q=\left(q_1, q_2\right)$ is the set of largest angular ranges for which every angle within the range is in the result of a reverse $\topk$ query, $q$.  Formally:

\begin{definition}\label{def:mrtop}
The response to a {\em maximal reverse $\topk$ query}, $q=\left(q_1,q_2\right)$, is the set of open intervals:
\begin{eqnarray*}
\mrtop(q)&=&\{\left(\theta_0\geq 0, \theta_1\leq\pi/2\right):\\
&&\theta_0\not\in RTOP(q)\wedge\\
&&\theta_1\not\in RTOP(q)\wedge\\
&&\forall\theta\in\left(\theta_0,\theta_1\right), \theta\in RTOP(q)\}.\\
\end{eqnarray*}
\end{definition}

Additionally to these problem definitions, we define here a number of concepts with which in the subsequent sections we assume the reader is familiar.  Specifically, we define here the {\em nullspace} of a vector, an {\em arrangement of lines}, the {\em $k$-skyband} of a set of points, and our novel concepts of {\em\dd} and {\em\dd\ contours}.

\begin{definition}\label{def:nullspace}
The {\em nullspace} of a vector $\vec{v}=\left<v_1,v_2\right>$ is the set of vectors orthogonal to $\vec{v}:\{\vec{u}:\vec{u}\cdot\vec{v}=0\}$.  In two dimensions, this is exactly the line $y=-\frac{v_1}{v_2}x$.  The {\em translated nullspace} of $\vec{v}$, given a positive real $\tau$, is the set of vectors $\{\vec{u}:\vec{u}\cdot\vec{v}=\tau\}$, or the line $y=\frac{\tau}{v_2}-\frac{v_1}{v_2}x$.
\end{definition}

\begin{definition}\label{def:arrangement}
An {\em arrangement} of a set of lines $\mathcal{L}$, denoted $\mathcal{A}_{\mathcal{L}}$, is a partitioning of $\mathbb{R}^2$ into cells, edges, and vertices.  Each {\em cell} is a connected component of $\mathbb{R}^2 \setminus \mathcal{L}$.  Each vertex is an intersection point of some two lines $l_1,l_2\in\mathcal{L}$.  An edge is a line segment between two vertices of $\mathcal{A}$.
\end{definition}

\begin{definition}\label{def:skyband}
Consider the set $\overline{\mathcal{S}_k}$ of tuples $(a_1, a_2)$ in $\mathcal{D}$ for which there are at least $k$ other tuples with higher values of both $a_1$ and $a_2$ (i.e., the set of points pareto-dominated by at least $k$ other points).  The {\em $k$-skyband} of $\mathcal{D}$, which we denote $\mathcal{S}_k$, is precisely the rest of $\mathcal{D}$: $\mathcal{D}\setminus\overline{\mathcal{S}_k}$.
\end{definition}

\begin{definition}\label{def:depth}
The {\em\dd} of a point $p$ within an arrangement $\mathcal{A}$, is the number of edges of $\mathcal{A}$ between $p$ and the origin.  That is to say, the depth of $p$ is the number of intersections between edges of $\mathcal{A}$ and $[\mathcal{O},p]$.  Similarly, the {\em\dd} of a cell of $\mathcal{A}$ is the \dd\ of every point within that cell.
\end{definition}

\begin{definition}\label{def:contour}
A {\em\dd\ contour} is the set of edges in an arrangement $\mathcal{A}_{\mathcal{L}}$ that have \dd\ exactly $k$.  We also refer to a {\em\dd\ contour} as the {\em $k$-polygon} of $\mathcal{L}$, because, as we show later, the contour is a closed, star-shaped polygon.
\end{definition}
\section{An Arrangement View}\label{sec:why}
The theme of this paper is to answer \mrtop\ queries with {\em logarithmic cost} by means of a data structure featuring a largely sequential data layout and inspired by geometric analysis of the problem.  In this section, we conduct that analysis and create the theoretical foundations for our correctness proof of our access methods in Section~\ref{sec:algorithm}.  

The approach taken in Vlachou\ea\ is to exploit the dominance relationship among points in $\mathcal{D}$.  The approach taken in Wang\ea\ is to compare all points in $\mathcal{D}$ to the query point $q$ in the dual space.  We take a very different approach.  We transform the dataset into an arrangement of lines and demonstrate that embedded in the arrangement is a critical polygon $\mathcal{P}_k$ which partitions $\mathbb{R}^2$ into points to include among and exclude from a \mrtop\ query result.  We show, too, that by applying the same transformation to the query to produce a line $l_q$, $\mrtop(q)$ is given precisely by the intersection of $l_q$ with the interior of $\mathcal{P}_k$.

An equally important contribution of this section is that we derive properties of $\mathcal{P}_k$ that are critical for proving later the asymptotic performance of our access method.

This section is thus divided into three subsections: the first describes the transformation of $\mathcal{D}$ into a set of $|\mathcal{D}|$ contours (Section~\ref{sec:contours}); the second derives important properties of $\mathcal{P}_k$ (Section~\ref{sec:properties}); and the third establishes the equivalence of the intersection test to the original \mrtop\ problem (Section~\ref{sec:equiv}).

\subsection{$\mathcal{A}_{\mathcal{L}}$ and Top-$k$ Rank Depth Contours}\label{sec:contours}
In this section we describe what is a \dd\ contour and how it is constructed from a relation, $\mathcal{D}$.  We illustrate how to construct the arrangement from $\mathcal{D}$ and how to interpret the arrangement as a set of contours.
First, in order to reason about $\mathcal{D}$ in terms of an arrangement, we need to represent each tuple as a line such that the relative positions of the lines with respect to a ray from the origin reflects their $\topk$ ranking.  This is precisely the property that is proferred by the translated nullspaces of each tuple, for any arbitrary real $\tau$.  

So, we convert the set of tuples (or, alternatively, vectors) $\mathcal{D}$ into a set of lines by transforming each tuple $v=(v_1, v_2)$ to the line $\overline{v}: y=\frac{\tau}{v_2}-\frac{v_1}{v_2}x$.  For a ray $r$ in any direction, we can show that:

\begin{lemma}\label{thm:bpProj}
If the depth of a point $v$ is less than the depth of a point $u$ in the direction of a ray $r$, then the rank of $v$ for a traditional, linear $\topk$ query $\vec{r}$ is better than that of $u$.
\end{lemma}
\begin{proof}
If the translated nullspace of $\vec{v}$ is closer to the origin than of $\vec{u}$ in the direction of $r$, then $\vec{v}\cdot\vec{r} =\tau= \vec{u}\cdot c\vec{r}$ for some $c>1$.  Therefore, $\vec{v}\cdot\vec{r} > \vec{u}\cdot\vec{r}$.
\qed
\end{proof}

In fact, we can make a stronger claim: the depth of a point $p$ is precisely its \dd\ for a query in the direction of $p$ if $p$ happens to correspond to a point on an edge of the arrangement.

\begin{corollary}\label{thm:ddGivesRank}${\ddf}(p) = \mathrm{rank}(\vec{p})$ for $TOP(\vec{p}).$
\end{corollary}

\begin{proof}
Let ${\ddf}(p)$ be $d$.  Then from the definition of \dd\ there are $d$ other baseplanes that will be sooner encountered by a ray emanating from $\origin$ in the direction of $p$.  From Lemma~\ref{thm:bpProj}, we know that each of these has a better rank than $p$, so the rank of $p$ is at best $d$.  Also, from Lemma~\ref{thm:bpProj} we can conclude that $p$ has a better rank than all those with translated nullspaces farther from the origin than that of $\vec{p}$, so the rank of $p$ is not greater than $d$, either.
\end{proof}

The $k$'th contour of an arrangement is the set of all edges at the same depth.  We wish to show that, in fact, the edges form a connected ring around the origin, thus forming a polygon.  In order for this to be true, we need to show that in any direction there is exactly one point on the contour, and that the points are all adjacent to each other.  This is the objective of the following three lemmata.

Firstly, to demonstrate connectedness, it is important that \dd\ is a monotone measure:

\begin{lemma}\label{thm:ddIsMonotone}
Top-$k$ rank depth increases monotonically with Euclidean distance from $\origin$ in any arbitrary direction.
\end{lemma}
\begin{proof}
Consider two points $p, q$ such that $p$ lies on the line segment $[\origin, q]$.  Every line in the arrangement that crosses $[\origin, p]$ also crosses $[\origin, q]$, so $\ddf(q)\geq \ddf(p)$.
\end{proof}

Secondly, we need to show that a cell of depth $i$ is unique in a given direction:

\begin{lemma}\label{thm:uniquePoly}
There is exactly one cell of depth $i$ in any given direction from $\origin$, for reasonably small $i$.
\end{lemma}

\begin{proof}
First, we show that there is at most one cell of depth $i$.
This follows from the definition of \dd.  Assume for the sake of contradiction that there are two disjoint cells, A and B, with depth $i$ in the same direction.  Without loss of generality, assume that A is nearer to $\origin$ than B.  Take some point $a\in A$.  Then, from the definition of \dd, we know that there are exactly $i$ lines crossing the line segment $[\origin, a]$.  Now consider some point $b\in B$.  Because A is nearer than B to $\origin$, clearly every line between $a$ and $\origin$ also crosses the line segment $[\origin, b]$.  So, too, must the upper boundary of A, since A and B are distinct.  But then there are at least $i+1$ lines crossing $[\origin, b]$, which contradicts that B is at depth $i$.

The assumption that $i$ is reasonably small is to guarantee that there are sufficiently many tuples in $\mathcal{D}$ that there are at least $i$ tuples to return for a traditional $\topk$ query.  This is enough to imply that there is an $i$-contour in every possible direction, so there must be at least one cell in our given direction at depth $i$, as well. 
\end{proof}

Thirdly, we can now show that, in fact, all cells of depth $i$ are connected and can thus form a contour:

\begin{corollary}\label{thm:connectedContours}All cells at the same \dd\ ($\leq k_{\mathrm{max}}$) are connected.
\end{corollary}
\begin{proof}
This follows from Lemma~\ref{thm:uniquePoly}, which implies that there are no discontinuities in the contour in any given direction.  Observe, too, that for any cell there must be an adjacent cell with the same depth at every corner.  The corners correspond to directions in which the incident translated nullspaces reverse order.  So, since the top translated nullspace becomes a bottom translated nullspace and vice versa, the depth does not change.\footnote{Strictly speaking, the vertex/corner itself is a discontinuity, as there is no point in that direction with exactly the right number of crossing line segments, but this is infinitesimal in size and we ignore the issue because we return open intervals anyway.}
\end{proof}

This is enough to establish that the $k$'th contour of the arrangement is precisely a star-shaped polygon:

\begin{theorem}\label{thm:polygonStarShaped}
A contour is a star-shaped polygon.
\end{theorem}

\begin{proof}
First, we know that the contour is connected and exists in every direction from $\origin$.  
Also, every point inside the polygon is visible from $\origin$, for if there were some point $p$ that were not visible, then an edge of the boundary would cross $[\origin, p]$.  However, this would imply that there are two cells at the same depth in the direction of $p$ from $\origin$, contradicting Lemma~\ref{thm:uniquePoly}.  
\end{proof}

Theorem~\ref{thm:polygonStarShaped} is quite important.  It establishes that we can represent $\mathcal{D}$ as a set of polygons with a unique depth $i$, each of which itself encodes the $i$'th ranked tuple for any possible traditional, linear $\topk$ query.  If there is only one value $k$ of interest, then the entire dataset can be represented just by one polygon.  In this next subsection, we show properties of the $k$-polygon, including bounds on its size, and in the following subsection describe how to use it in order to address the main question of this paper, \mrtop\ queries.

\subsection{Properties of $\mathcal{P}_k$}\label{sec:properties}
In order to be able to use $\mathcal{P}_k$ as a data structure, we have to evaluate properties of the polygon in order to evaluate asymptotic performance.  As we will detail in the next section, our data structure will be a representation of $\mathcal{P}_k$, so the number of edges and vertices in the polygon influences our access time.  

Also, to improve performance, our data structure includes a convex approximation of $\mathcal{P}_k$ (specifically the convex hull), and understanding the implications of this approximation is also important.  

Thirdly, we approximate the dataset $\mathcal{D}$ by $\mathcal{S}_k$, so understanding the implications of this approximation is clearly important, as well.

Gathering this understanding is the intent of these next three lemma.  Specifically, they answer these three questions in order:

\begin{lemma}\label{thm:arrangeBound}
An arrangement of $m$ lines can produce contours at \dd\ $i$ with no more than $\mathcal{O}(m)$ edges.
\end{lemma}

\begin{proof}
Note from Theorem~\ref{thm:ansIsIntersection} that for each line $l$ derived from a tuple $v$, the edges it contributes to the $k$'th contour are precisely the answer to a \mrtop\ query of $v$ on $\mathcal{D}\setminus\{v\}$.  From Proposition~\ref{thm:twoSegs}, we know this can consist of at most two disjoint angular intervals; therefore, $l$ can contribute at most two edges to the $k$'th contour.
\end{proof}

\begin{lemma}\label{thm:concavityBound}
A concave region between vertices of the convex hull of the $k$'th contour's upper boundary can have at most $2k-1$ vertices.
\end{lemma}

\begin{proof}
Notice that vertices of the convex hull of the contour's upper boundary are themselves at depth $k-1$.  Consider two such vertices, $v_i$, $v_j$, delimiting a concave region.  Any line that passes neither under $v_i$ nor under $v_j$ and is orthogonal to some non-zero vector from $\origin$ cannot pass through the concave region's face, so the face is defined by at most $2k$ lines.  This is, in fact, an arrangement, so Lemma~\ref{thm:arrangeBound} implies the bound on the number of cells in that arrangement that could possibly be at depth $k$ and thus contribute a vertex to the concave region's boundary.
\end{proof}

\begin{lemma}
Only tuples in $\mathcal{S}_k$ can form part of the $k$-polygon.
\end{lemma}
\begin{proof}
Tuples that are not among $\mathcal{S}_k$ are, by definition, among $\overline{\mathcal{S}_k}$.  However, the tuples of $\overline{\mathcal{S}_k}$ are those dominated by at least $k$ other tuples.  In order words, regardless of the traditional, linear $\topk$ query issued, there are at least $k$ better ranked tuples.  Consequently, the $k$-polygon, which encodes the $k$'th ranked tuples for all possible traditional, linear $\topk$ queries, clearly does not contain the tuples of $\overline{\mathcal{S}_k}$ in any direction.
\end{proof}

\begin{figure*}[bt]
 \begin{centering}
  \includegraphics[scale=1.25, clip=true, trim=0 8in 6cm 0]{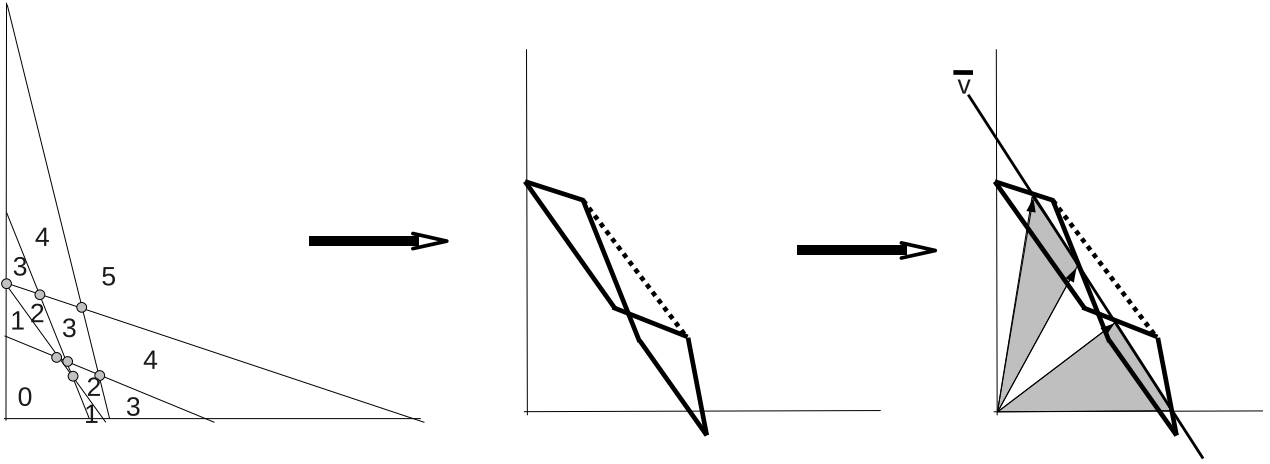}
  \caption{\label{fig:contours}An arrangement labelled with \dd; the $2$'nd contour, zoomed in, with its convex hull displayed by the dashed line; and the (coloured gray) result for the reverse top-2 query for the vector $\vvec=\left<5, 5/2\right>$ (whose baseplane is $y=2(\tau-4x)/5$).}
 \end{centering}
\end{figure*}

\subsection{A Transformed \mrtop\ Query}\label{sec:equiv}
In the previous subsections we have demonstrated that a star-shaped polygon (the $k$-polygon) can encode the $k$'th best ranked tuple for all query directions.  In this section, we demonstrate how to use the $k$-polygon for \mrtop\ queries.

First, recall that the arrangement of lines was produced by transforming each tuple in $\mathcal{D}$ to its translated nullspace, given some fixed but arbitrary $\tau$.  Here, we prove that applying the same transformation to a query $q$ to produce a line $l_q$ and intersecting $l_q$ with the interior of $\mathcal{P}_k$ yields the directions in which $q$ is among the result set of traditional, linear $\topk$ queries:

\begin{theorem}\label{thm:ansIsIntersection}
The response to a \mrtop\ query, given query vector $\vec{q}=\left<q_1,q_2\right>$, is the component of $\overline{q}: y=\frac{\tau}{q_2}-\frac{q_1}{q_2}x$ which intersects the interior of $\mathcal{P}_k$.
\end{theorem}

\begin{proof}
Recall from Theorem~\ref{thm:ddGivesRank} that the $k$'th contour corresponds exactly to the vectors of rank $k$ and also from Lemma~\ref{thm:ddIsMonotone} that the contours increase in rank monotonically.  Therefore, if we constructed a new arrangement which also contained $\overline{q}$, the components of $\overline{q}$ which lay outside the $k$'th contour would be directions in which the rank of $q$ is greater than $k$.  The inverse of this is the solution to the \mrtop\ query.
\end{proof}

Consequently, it suffices to develop algorithms for solving the problem of identifying the segments of $\overline{q}$ which lie inside the $k$'th \dd\ contour in order to solve the \mrtop\ problem.  An illustration of this is provided in Figure~\ref{fig:contours}.

A final note regarding the properties of $\mathcal{P}_k$ is that:

\begin{proposition}\label{thm:twoSegs}
The result in two dimensions of a \mrtop\ query consists of at most two continuous intervals.
\end{proposition}



\section{Efficiently Answering \mrtop\\Queries}\label{sec:algorithm}
Having established the theoretical foundations in the previous section, we present here our index structure and access method.  A key insight that we derived earlier is that the \mrtop\ response to $q$ is the intersection of $l_q$ with the interior of $\mathcal{P}_k$.  Fittingly, then, our index structure is a representation of $\mathcal{P}_k$ and our access method is an efficient means of retrieving from the index the intersection points of $l_q$ with $\mathcal{P}_k$.  First we give a high-level overview of our algorithms and data structure and then present the precise details in the upcoming subsections.

Not just any representation of $\mathcal{P}_k$ will suffice: it has to facilitate the efficiency of the access method.  We accomplish this by creating a binary search procedure to identify the intersections of $l_q$ with the convex hull of $\mathcal{P}_k$.  This leads to an efficient access method because we established Theorem~\ref{thm:concavityBound}.  We also aim to achieve a very sequential data layout to improve read times.  So, we have developed a data structure consisting of one ordered list of the vertices of the convex hull of $\mathcal{P}_k$ and one ordered list of ordered lists of $\mathcal{P}_k$ vertices not on the convex hull.  We describe the index structure in Section~\ref{sec:representation}. 

Algorithmically speaking, there are two considerations.  Of foremost importance is how to efficiently query the index structure, given $l_q$ (Section~\ref{sec:querying}).  The second consideration is how to efficiently construct (Section~\ref{sec:construction}) the index structure described above.  Let us begin by addressing the first.

The idea is to exploit properties of the problem.  Our binary search to discover the intersection points of $l_q$ with a convex polygon is of logarithmic cost.  Furthermore, given the intersection points of $l_q$ with the convex hull of $\mathcal{P}_k$, we can find the exact intersection of $l_q$ with $\mathcal{P}_k$ by comparing it with every edge ``shaved off'' by that convex hull edge.  By Theorem~\ref{thm:concavityBound}, we know there are most $\mathcal{O}(k)$ such edges.  Because of our sequential layout, a direct comparison to each of these $\mathcal{O}(k)$ edges is affordable.

Our construction algorithm is a plane sweep algorithm.  We sweep radially from the positive $x$-axis to the positive $y$-axis, maintaining a list of all the lines in sorted order with respect to their intersection points on the sweep line.  At any given moment during the plane sweep, the $k$'th line in the list is the edge of the $k$-polygon.  So, identifying the $k$-polygon is equivalent to identifying all the points at which the $k$'th line in that list changes.  These points are the vertices of the $k$-polygon.  Maintaining the convex hull of the polygon is fairly straight-forward if one maintains convexity as an invariant throughout the sweep.

The expense of this algorithm is dominated by initially sorting all the lines with respect to their intersection points with the $x$-axis.  We improve upon this by recognising that only tuples of the $k$-skyband are relevant.  So, at the cost of an extra sequential scan, we approximate the $k$-skyband with perfect recall (i.e., ensure every true positive is in the approximation) and then construct $\mathcal{P}_k$ from that approximation, rather than from all of $\mathcal{D}$.  

The approximation method exploits the work we have already done in this paper.  We note that if a tuple is in the $k$-skyband of $\mathcal{D}$, then it must be in the $k$-skyband of any subset of $\mathcal{D}$.  So, we build our index structure on $2k$ selected tuples from $\mathcal{D}$ and then include in our approximation any tuples which have non-null \mrtop\ query responses on that small index structure.

Together, these algorithms and this data structure gives Theorem~\ref{thm:2dAsymptotics}:

\begin{theorem}\label{thm:2dAsymptotics}
The two dimensional \mrtop\ problem can be solved using $\mathcal{O}(\log{n}+k)$ query time with an index that requires $\mathcal{O}(n)$ disk space.
\end{theorem}

Under the practical assumption that $k$ is constant or $\mathcal{O}(\log{n})$, the above theorem implies that query cost is $\mathcal{O}(\log{n})$.

\subsection{The $k$-Polygon Index Structure}\label{sec:representation}
Facilitating logarithmic query time of the index largely depends on how the data is represented.  Our idea is to exploit Theorem~\ref{thm:concavityBound} in our representation.  Let $\mathcal{H}$ denote the set of vertices of the convex hull of a $k$-polygon, $\mathcal{P}_k$.  We maintain two arrays, which we collectively refer to as the {\em dual-array representation} of $\mathcal{P}_k$.  The first, which we call the {\em convex hull array}, contains the $|\mathcal{H}|$ vertices of $\mathcal{H}$, ordered anti-clockwise from the positive $x$-axis.  The second array, which we call the {\em concavity array}, is of size $|\mathcal{H}|-1$.  The $i$'th entry contains a sequential list of the up to $2k-1$ vertices of the $k$-polygon between the $i$'th and $(i+1)$'st vertices of $\mathcal{H}$.

\subsection{Construction of the $k$-Polygon}\label{sec:construction}
Although Section~\ref{sec:properties} suggests how to determine the $k$-polygon of $\mathcal{D}$ by first constructing an arrangement of lines and then extracting from it all the edges at a \dd\ of $k$, here we describe a much more efficient algorithm.  
The key insight is that the only tuples that could form part of the $\mathcal{P}_k$ are those among the $k$-skyband of $\mathcal{D}$.  So, we approximate the $k$-skyband with perfect recall, sort those lines based on their x-intercept (the dominating cost of the algorithm), and introduce a radial plane sweep algorithm to build the polygon index.

\subsubsection*{$k$-Skyband Approximation}\label{sec:skybandApprox}
The important consideration in our $k$-skyband approximation is that perfect recall is critical.  Otherwise, we may miss a line that forms part of the $k$-contour.  We exploit the insight that the $k$ best lines with respect to each axis form a contour relatively close to the real contour, and that if a tuple is in the $k$-skyband, it clearly must be in the $k$-skyband of any subset of the data.  Thus, the approximation algorithm proceeds by quickly determining the $\leq2k$ lines as above, constructing a contour from them, and determining which lines in $\mathcal{D}$ have non-null \mrtop\ query answers on the approximate contour.  See Algorithm~\ref{alg:skybandApprox}.

\begin{algorithm}[h]
 \caption{\label{alg:skybandApprox}Approximating the $k$-skyband of $\mathcal{D}$}
 \begin{algorithmic}[1]
  \STATE \textbf{Input}: $\mathcal{D}$; $k$
  \STATE \textbf{Output}: $\mathcal{S}\subseteq\mathcal{D}$, the tuples that form the $k$-skyband of $\mathcal{D}$, plus potentially some false-positives
  \STATE Initialise $\mathcal{S}$, an empty set of tuples
  \STATE Let $\mathcal{X}$ denote the $k$ tuples in $\mathcal{D}$ with the highest values for attribute $x$
  \STATE Let $\mathcal{Y}$ denote the $k$ tuples in $\mathcal{D}$ with the highest values for attribute $y$
  \STATE Construct $\mathcal{P}_{\mathcal{X}\cup\mathcal{Y}}$, the $k$-polygon index on the set $\mathcal{X}\cup\mathcal{Y}$ using Algorithm~\ref{alg:planeSweep}.
  \FORALL{$p\in\mathcal{D}$}
    \IF{$l_p$ intersects the interior of $\mathcal{P}_{\mathcal{X}\cup\mathcal{Y}}$ or $p\in\mathcal{X}\cup\mathcal{Y}$}
      \STATE Add $p$ to $\mathcal{S}$
    \ENDIF
  \ENDFOR
  \STATE Free $\mathcal{X}$ and $\mathcal{Y}$.
  \STATE RETURN $\mathcal{S}$.
 \end{algorithmic}
\end{algorithm}

\subsubsection*{Radial Plane Sweep}\label{sec:planeSweep}
We construct a contour from a set of lines using a radial plane sweep.  The idea is to traverse the set of intersection points in angular order, maintaining a sorted list of the lines.  In this way, we build the contour incrementally from the positive $x$-axis towards the positive $y$-axis.  Traversing in this order also allows us to maintain convexity of the contour as we go.  Like most plane sweeps, a primary advantage is that we need only look at intersection points between two lines after they become neighbours.  If this does not occur between the sweep line and the positive $y$-axis, then we need not consider the intersection point at all.  Algorithm~\ref{alg:planeSweep} offers the details of the sweep algorithm.

\begin{algorithm}[h]
 \caption{\label{alg:planeSweep}Building $\mathcal{P}_k$ from a $k$-skyband approximation}
 \begin{algorithmic}[1]
  \STATE \textbf{Input}: $\mathcal{L}$, an array of lines sorted by ascending $x$-intercept; $k$
  \STATE \textbf{Output}: A dual-array representation of $\mathcal{P}_k$
  \STATE Initialise an empty array $\mathcal{H}$ for convex hull vertices
  \STATE Initialise an empty array of lists $\mathcal{C}$ for concavities
  \STATE Initialise $\mathcal{I}$ as a priority queue containing the $|\mathcal{L}|-1$ intersections of neighbouring lines in $\mathcal{L}$, sorted by angle from the positive $x$-axis, discarding those $<0$.
  \WHILE{$\mathcal{I}$ is not empty}
    \STATE Pop next intersection $i\in\mathcal{I}$
    \STATE Let $l_{left}$ and $l_{right}$ be the lines intersecting at $i$.
    \IF{$l_{left}=\mathcal{L}_{k-1}$ or $l_{right}=\mathcal{L}_{k-1}$}
      \STATE Add $i$ to $\mathcal{H}$
      \IF{$\exists h\in\mathcal{H}: \mathrm{slope}([h,i])<\mathrm{slope}([h,h+1])$}
        \STATE Add to $\mathcal{C}_h$ all vertices between $h$ and $i$.
        \STATE Remove all vertices between $h$ and $i$ from $\mathcal{H}$ and from $\mathcal{C}_j, \forall j\neq h$.
      \ENDIF
    \ENDIF
    \STATE Swap $l_{left}$ and $l_{right}$ in $\mathcal{L}$
    \STATE Add to $\mathcal{I}$ the intersection of $l_{left}$ with its new neighbouring line and the intersection of $l_{right}$ with its new neighbouring line, provided they are at angles greater than that of $i$ and in the positive quadrant
  \ENDWHILE
  \STATE Free $\mathcal{I}$.
  \STATE RETURN $\mathcal{H}$ and $\mathcal{C}$.
 \end{algorithmic}
\end{algorithm}

\subsection{Querying the $k$-Polygon Index}\label{sec:querying}
Here we present how to query our $k$-polygon index to determine the segments of a line $l_q$ that are strictly contained within the interior of the $k$-polygon, $\mathcal{P}_k$.  The algorithm (Algorithm~\ref{alg:querying}) is a binary search on the convex hull of the polygon, proceeded by a sequential scan of $\mathcal{O}(k)$ edges of $\mathcal{P}_k$.  The recursion is based on the slope of $l_q$ compared to the convex hull of $\mathcal{P}_k$ at the recursion point.  

\begin{algorithm}[h]
 \caption{\label{alg:querying}Querying a dual-array $k$-polygon, $\mathcal{P}_k$}
 \begin{algorithmic}[1]
  \STATE \textbf{Input}: Dual-array representation of $\mathcal{P}_k$, line $l_q$, start/end indexes.
  \STATE \textbf{Output}: Intersection points of $l_q$ with $\mathcal{P}_k$
  \IF{$end-start=2$}
  	\STATE Traverse the $\mathcal{O}(k)$ list in the concavity array at position $start$, returning any intersections with $l_q$.
  	\STATE RETURN.
  \ENDIF  
  \STATE Compute midpoint vertex of $\mathcal{H}$ at $\frac{end-start}{2}+start$.
  \IF{$l_q$ passes above midpoint}
	\IF{slope of $l_q$ is less than slope of [midpoint-1, midpoint]}
		\STATE Recurse on lower half with end=midpoint
	\ELSIF{slope of $l_q$ is greater than slope of [midpoint, midpoint+1]}
		\STATE Recurse on upper half with start=midpoint
	\ENDIF
  \ELSE 
  	\IF{$l_q$ passes above vertex at position start}
  	  \STATE Recurse on lower half with end=midpoint
  	\ENDIF
  	\IF{$l_q$ passes above vertex at position end}
  	  \STATE Recurse on upper half with start=midpoint
  	\ENDIF
  \ENDIF
 \end{algorithmic}
\end{algorithm}

\subsection{Asymptotic Performance}\label{sec:asymp}
Earlier we stated the asymptotic performance of our algorithms.  Here, now, we have the tools to prove that theorem.  The basic idea is that a line can only intersect a convex shape in two locations and for each of those intersection points, the cost of a face traversal is bounded.

\begin{proof}[of Theorem~\ref{thm:2dAsymptotics}]
First, note that a line can only intersect the boundary of a convex polygon in at most two points, so the binary search tree traversal need follow at most two paths.  Recall from Lemma~\ref{thm:arrangeBound} that each contour contains at most $n$ cells, and thus the convex hull contains at most $n-1$ edges.  From Lemma~\ref{thm:convexHullLgN}, the binary search requires $\mathcal{O}(\log{n})$.  For each of the two intersection points found, we traverse the corresponding face sequentially.  From Lemma~\ref{thm:concavityBound}, each of these faces contains $\mathcal{O}(k)$ edges and we know that finding the intersection (or, equivalently, ascertaining the non-intersection) of two two-dimensional line segments requires constant time.

Since the search is run independently of and its cost dominates the cost of the face traversals, and since we assume $k$ is $\mathcal{O}(\log{n})$, the entire query procedure is $\mathcal{O}(\log{n})$.

Regarding the space requirements, Lemma~\ref{thm:arrangeBound} implies that polygon itself can contain at most $\mathcal{O}(n)$ vertices.  Because each vertex could appear at most twice in the data structure (one on the convex hull and once in a single concavity), and because the data structure is, simply, the vertices of the $k$-polygon, the disk space required by the data structure is $\mathcal{O}(n)$.
\end{proof}

\begin{lemma}\label{thm:convexHullLgN}
The intersection of the query line with the convex hull can be determined in $\mathcal{O}(\log{n})$.
\end{lemma}
\begin{proof}
The intersection algorithm proceeds by binary search.  First, find the middle vertex $v_{n/2}$ and determine whether the query line passes above or below it.  If above then recurse left if the query line has shallower slope than edge $(v_{n/2}, v_{n/2+1})$.  Recurse right if the query line has steeper slope than edge $(v_{n/2-1},v_{n/2})$.  Because edge $(v_{n/2}, v_{n/2+1})$ is shallower than edge $(v_{n/2-1},v_{n/2})$, at most one recursion direction can be followed.

If, instead, the query line passes below $v_{n/2}$, then it is inside the contour (if in the correct quadrant at all).  To find the intersection points, recurse left if the query line passes above $v_{n-1}$.  Recurse right if the query line passes above $v_0$.  It is possible that both conditions are true, but this can only occur once, because the truth of the condition implies an intersection point and a straight line has at most two intersection points with a convex polygon.  Therefore, the binary search follows at most two distinct paths.
\end{proof}

\section{Experimental Evaluation}\label{sec:exper}
Until now, the focus of this paper has been on proving the correctness and asymptotic performance of our approach to indexing for monochromatic reverse $\topk$ queries.  Here, we pursue a different direction, examining the question of performance in more detail through experimentation.  In particular, we seek to address two questions.  As we showed earlier, if the size of the $k$-polygon is $|DS|$ and the size of its convex hull is $|CH|$, then the query cost of our index is $\mathcal{O}(k + \log\ |CH|)$.  So, a natural question is what a typical value of $|DS|$ and of $|CH|$ might be.  This is also important because it indicates how much space the data structure will consume on disk once built.  The second question is that of raw performance: in how much time can the index be built and, later, be queried?  To contextualise these performance numbers, we compare the performance of our index to that of repeatedly executing the non-indexed-based algorithms of Vlachou\ea\ and of Wang\ea

\subsection{Experimental Setup}
For the experiments, we implemented and optimised the algorithms of Vlachou\ea, of Wang\ea, and of this paper (Chester\ea) in C and compiled our implementations with the GNU C compiler 4.4.5 using the {\em -O6} flag.  Our implementations of Vlachou\ea\ and of Wang\ea\ do not produce maximal intervals, although, naturally, ours does.  We ignore the cost of outputting the response, because this is moreorless the same for each algorithm.  On the other hand, each interprets the tuples differently, so we do include in the measurements the cost of reading the input files.

We ran the experiments on a machine with an AMD Athlon processor with four 3GHz overclocked cores and 8GB RAM, running Ubuntu.  The timings were calculated twice, once using the linux {\em time} command and once using the {\em gnu profiler} by compiling with the {\em -pg} flag.

\paragraph{Datasets}
We use the regular season basketball player statistics from {\em databasebasketball.com} and generate five different datasets with which to perform our experiments by projecting combinations of two attributes.  The attribute combinations are chosen to diversify the degree of (anti-)correlation based on intuitive reasoning about the attributes.  In particular we study the following pairs: (Points, Field Goals Made), (Defensive Rebounds, Blocks), (Personal Fouls, Free Throw Attempts), (Defensive Rebounds, Assists), (Blocks, Three Pointers Made).  A traditional $\topk$ query on each pair is equivalent to asking for the $k$ best player seasons according to a given blend of the skills. (Note that for the first pair and a query (1,0), for example, Wilt Chamberlain would appear several times, once for each of his sufficiently high scoring seasons.)

We reserve the most recent season, 2009, as a set of $578$ query points and use the other seasons, 1946-2008, as the dataset of $21383$ tuples.  As such, each monochromatic reverse $\topk$ query is equivalent to asking, ``for which blends, if any, of the given two skills was this particular player's performance this season ranked in the $\topk$ all-time?''  This contrasts to traditional analysis of basketball data which would look only at the axes at the detriment of {\em rounded} players.  

In terms of preprocessing on the data, we elect not to remove multiple tuples for players who played on more than one team in a given year.  We scale the data to the range $(0,1]$ by adding $1$ to each value and dividing by the largest value for each attribute (plus one), so that the attributes are comparable in range.  Also, we slightly perturb the data so as not to violate our general position assumption by adding $10^{-8}$ to each duplicated scaled value until all values are unique for each attribute.  To construct our index and to process incoming queries, we assume a value of $\tau=.5$.\footnote{This choice really is arbitrary within reason.  We tried several values in the range [.25, 1.5] without any effect on the output.}

\subsection{Experimental Results}


\begin{figure*}[htb]
 \centering
 \includegraphics[scale=.75, clip=true, trim=.5in 3.75in 0 3.5in]{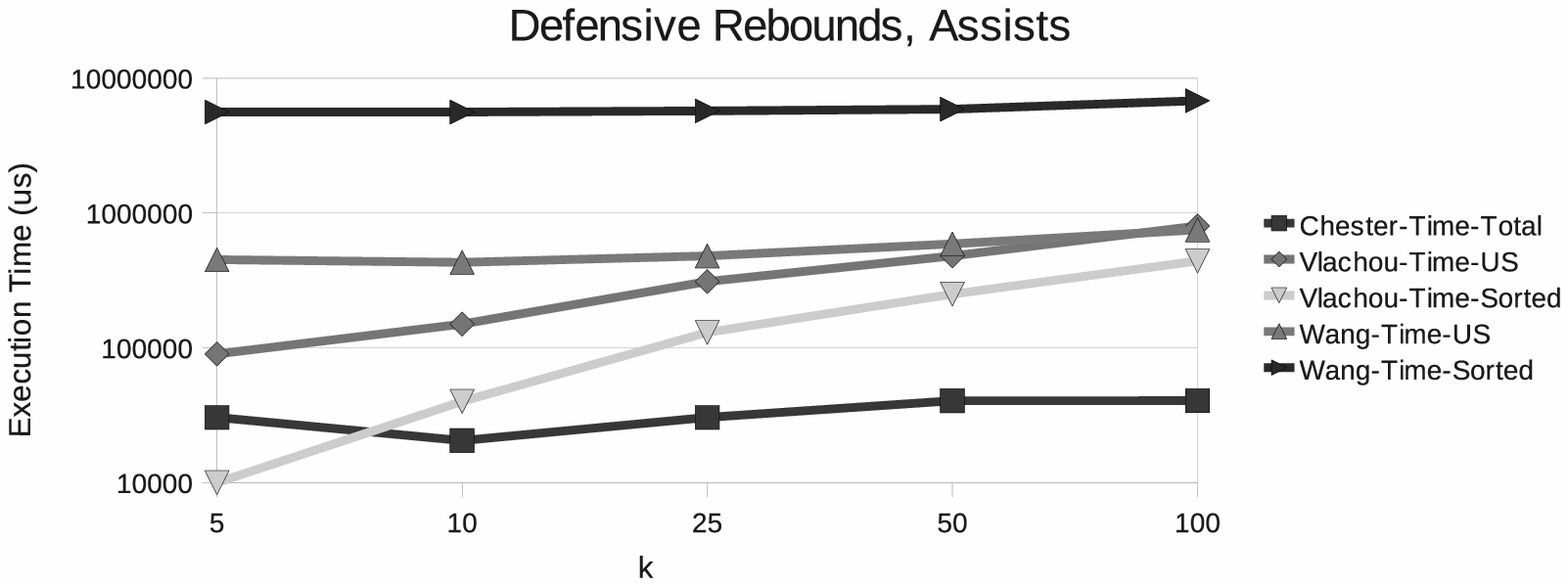}
 \caption{\label{fig:res-drebs-asts}Execution time for the implemented algorithms on the batch of 578 queries comprised of 2009 basketball statistics, using the statistics from 1946-2008 as a dataset.  {\em Defensive Rebounds} is regarded as the $x$-attribute and {\em Assists} is regarded as the $y$-attribute.  This is meant to reflect anticorrelated attributes, but the data appears to be more correlated. }
\end{figure*}

\begin{figure*}[tb]
 \centering
 \includegraphics[scale=.75, clip=true, trim=.5in 3.5in 0 3.5in]{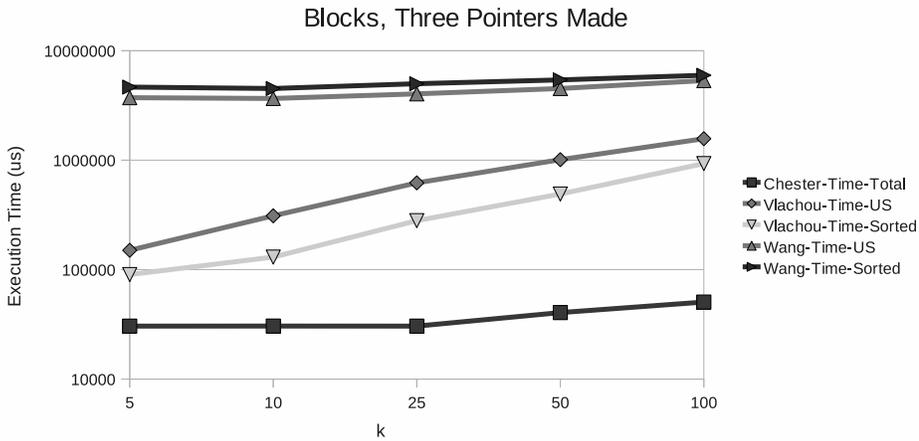}
 \caption{\label{fig:res-blks-tpm}Execution time for the implemented algorithms on the batch of 578 queries comprised of 2009 basketball statistics, using the statistics from 1946-2008 as a dataset.  {\em Blocks} is regarded as the $x$-attribute and {\em Three Pointers Made} is regarded as the $y$-attribute.  This is meant to reflect anticorrelated attributes.}
\end{figure*}

\begin{figure*}[bt]
 \centering
 \includegraphics[scale=.75, clip=true, trim=.5in 3.75in 0 3.5in]{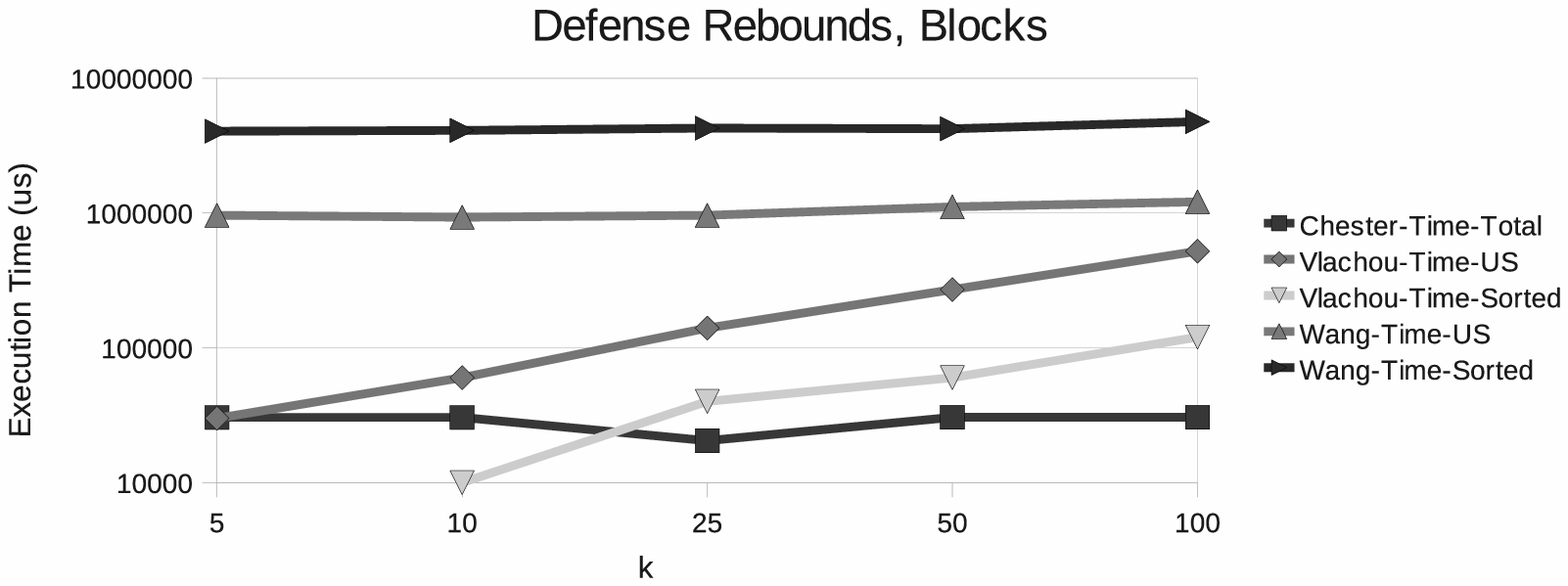}
 \caption{\label{fig:res-drebs-blks}Execution time for the implemented algorithms on the batch of 578 queries comprised of 2009 basketball statistics, using the statistics from 1946-2008 as a dataset.  {\em Defensive Rebounds} is regarded as the $x$-attribute and {\em Blocks} is regarded as the $y$-attribute.  This is meant to reflect correlated attributes, but the data appears to be more anticorrelated. }
\end{figure*}

One intention of these experiments was to illustrate how construction and query time varied for our algorithm with respect to $k$ and different attribute combinations.  However, the execution time of our algorithm is pretty much constant across values of $k$ and choices of attributes on the basketball dataset.  In fact, across all experiments the construction time has an average duration of $34$ms with a standard deviation of $9.1$ms.  The query time has an average duration of $480\mu$s with a standard deviation of $22\mu$s.  The total time for construction and querying averages $35$ms and has a standard deviation of $8.8$ms.  Figures~\ref{fig:res-drebs-asts}-\ref{fig:res-drebs-blks} illustrate the total execution times for the three algorithms on three of the attribute combinations.\footnote{We omitted results for the pair (Points, Field Goals Made) because it was very similar to the results for the pair (Defensive Rebounds, Blocks) and the pair (Personal Fouls, Free Throw Attempts) because it was very similar to the results for the pair (Defensive Rebounds, Assists), just with a larger separation between the lines.} We observed that the algorithms of Vlachou\ea\ and of Wang\ea\ are rather sensitive to the sortedness of the input data, so we report their performances both for when the data is presorted by $y$ value and when that presorted file is randomized with the linux command {\em sort -R}.

The other primary intention of the experiments was to evaluate the size of our data structure, particularly since it has a strong effect on the query time.  We show the contours generated for $k=[1,4]$ in Figures~\ref{fig:sage_ind}~and~\ref{fig:sage_ac} for two of the attribute combinations.\footnote{We omit figures for the other three combinations because the contours are too close together to interpret easily.}  We illustrate in Figures~\ref{fig:ds-blks-tpm}~and~\ref{fig:ds-pts-fgm} how the size of the data structure varies with $k$ on the attribute pairs (Personal Fouls, Free Throws Attempted) and (Points, Field Goals Made), respectively.  The former pair produced the largest data structures and the latter, the smallest.  The other three experiments all exhibited very similar behaviour, with the convex hull remaining relatively constant and the total size growing linearly with $k$, and magnitudes between these examples.

\begin{figure}
 \centering
  \includegraphics[scale=.45, clip=true, trim=0 2.45in 0 2.5in]{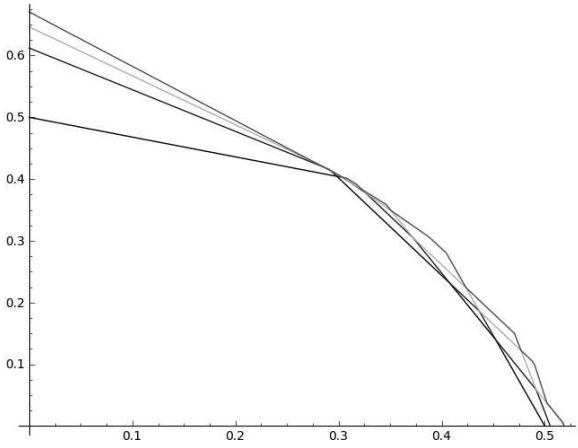}
  \caption{\label{fig:sage_ind}The first four contours derived on the basketball dataset with {\em personal fouls} as the $x$ attribute and {\em free throws attempted} as the $y$ attribute.}
\end{figure}

\begin{figure}
 \centering
  \includegraphics[scale=.45]{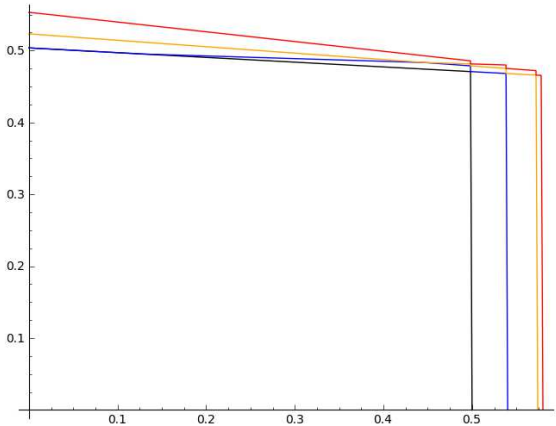}
  \caption{\label{fig:sage_ac}The first four contours derived on the basketball dataset with {\em blocks} as the $x$ attribute and {\em three pointers made} as the $y$ attribute.}
\end{figure}

\begin{figure}
 \centering
  \includegraphics[scale=.7, clip=true, trim=1.5in 3.5in 0in 3.5in]{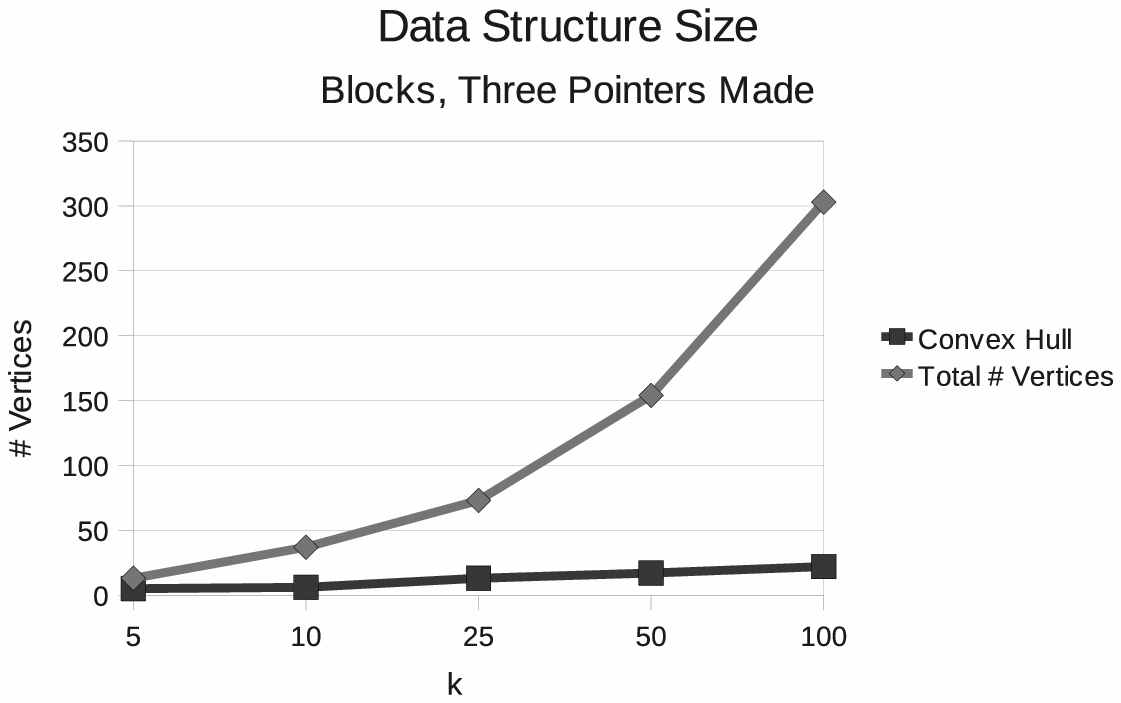}
  \caption{\label{fig:ds-blks-tpm}The size of the contours derived on the basketball dataset with {\em personal fouls} as the $x$ attribute and {\em free throws attempted} as the $y$ attribute, shown as a function of $k$.}
\end{figure}

\begin{figure}
 \centering
  \includegraphics[scale=.7, clip=true, trim=1.5in 3.5in 0in 3.5in]{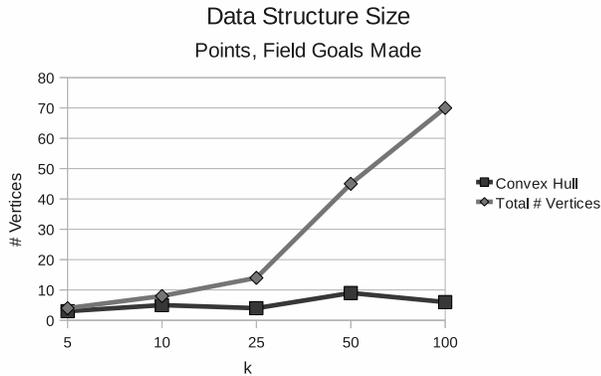}
  \caption{\label{fig:ds-pts-fgm}The size of the contours derived on the basketball dataset with {\em personal fouls} as the $x$ attribute and {\em free throws attempted} as the $y$ attribute, shown as a function of $k$.}
\end{figure}

\subsection{Discussion}
Overall, our indexing does quite well, with a query cost slightly less than $1\mu$s per query, independent of $k$, typically three to four orders of magnitude faster than the other two algorithms.  In fact, our algorithm in most cases runs one or two orders of magnitude faster, even with the construction cost included.  This implies that, while the purpose of this technique is to support an indexing scenario, the index construction is sufficiently fast to render it feasible in non-indexing scenarios, too.

That the query time for the index does not vary much is not surprising in light of the results of the data structure size analysis.  We see from Figures~\ref{fig:ds-blks-tpm}~and~\ref{fig:ds-pts-fgm} that the convex hull is consistently under forty vertices, and, from Lemma 4.8, we know that $|DS|\leq (2k-1)|CH|$, which explains the growth of $|DS|$ with respect to $k$.  

Since the query cost of our index is thus $\mathcal{O}(\log\ 40 + k)$, our performance is quite realistic.  The speed of the construction is more surprising, on the other hand, since its cost is proportional to the number of intersection points in the positive quadrant of the dual data lines.  This could be related to the choice of dataset because there could be a strong stratification of the lines such that they do not intersect in the positive quadrant given how strongly the statistics are influenced by playing time.  Nonetheless, quick construction time is not the primary objective of the index, anyway.

It is worth noting that there are a few instances in which the algorithm of Vlachou\ea\ outperforms our index for low values of $k$, especially on sorted data. (This is especially noticeable in Figure~\ref{fig:res-drebs-blks}, wherein the algorithm accumulates no time at the granularity of the {\em time} command.)  This can be easily explained because as soon as $k$ points are seen in the dataset that dominate the query, a null result can be reported and the Vlachou\ea\ algorithm can be halted.  When $k$ is low, this is substantially more likely.  When the data is sorted, these dominating points will be among the first seen.

A last observation for discussion is the difference in the shape of the contours produced by different data distributions (Figures~\ref{fig:sage_ind}~and~\ref{fig:sage_ac}).  The exaggerated slopes in the former, contrasted against the intricate weaving patterns in the latter, reflect the anticorrelatedness of the underlying data.  Insight into the shapes of contours could be a grounds for future work on $k$-polygon construction algorithms.

\section{Conclusion and Future Work}\label{sec:conclusion}
In this paper we introduced an index structure to asympotically improve query performance for reverse $\topk$ queries.  We approach the problem novelly by representing the dataset as an arrangement of lines and demonstrating that embedded in the arrangement is a critical $k$-polygon which encodes sufficient knowledge to respond to reverse $\topk$ queries.  In particular, we show that by applying the same transformation to the query tuple to produce a query line $l_q$, we can retrieve the response to the reverse $\topk$ query on $q$ by intersecting $l_q$ with the interior of the $k$-polygon.

We derive geometric properties of the problem to bound the query cost and size of our data structure as $\mathcal{O}(\log{n})$ and $\mathcal{O}(n)$, respectively.  We also conduct an experimental analysis to augment our theoretical analysis and demonstrate both that our algorithm significantly outperforms literature as the number of queries increases and that our data structure requires little disk space.

We believe this work can be extended in many directions.  Particularly, we feel that our index structure could lead to improved execution times for {\em bichromatic} reverse $\topk$ queries, as well.  Also, our geometric analysis of the problem space offers insight into traditional, linear $\topk$ queries, and exploring whether some of this research can be applied in that context is an interesting avenue.  Thirdly, still the difficult question of higher dimensions, and especially the question of how to represent solutions to higher dimension \mrtop\ queries, is open.

\bibliographystyle{abbrv}
\bibliography{mrtop-refs}

\end{document}